\documentclass[letterpaper,11 pt]{article} 
\usepackage{amsmath,amsfonts,amssymb,color}
\usepackage{mathtools}
\usepackage{dsfont}
\usepackage[dvipsnames]{xcolor}
\definecolor{mygreen}{rgb}{0.0, 0.5, 0.0}
\definecolor{winered}{rgb}{0.8,0,0}
\definecolor{myblue}{rgb}{0,0,0.8}
\usepackage{tikz}
\usepackage{pmat}
\usepackage{amsthm}
\usepackage{empheq}

\newtheorem{problem}{Problem}

\newtheorem{definition}{Definition}
\newtheorem{theorem}{Theorem}
\newtheorem{lemma}{Lemma}
\newtheorem{proposition}{Proposition}
\newtheorem{corollary}{Corollary}

\newtheorem{assumption}{Assumption}

\newtheorem{example}{Example}

\usepackage{algorithm}
\usepackage{algpseudocode}
\usepackage{epsfig}
\usepackage{array}
\usepackage{multirow}
\usepackage{epstopdf}
\usepackage{tikz}
\usepackage{relsize}
\usetikzlibrary{shapes,arrows}
\usepackage[hidelinks]{hyperref}
\hypersetup{
    colorlinks=true,
    linkcolor={winered},
    citecolor={mygreen}
}
\usepackage{cite}
\usepackage[margin=1in]{geometry}
\title{\LARGE \bf Distributed Hypothesis Testing and Social Learning in Finite Time with a Finite Amount of Communication}
\author{Shreyas Sundaram and Aritra Mitra
\thanks{The authors are with the School of Electrical and Computer Engineering at Purdue University.   Email: {\tt \{sundara2,mitra14\}@purdue.edu}. This work was supported in part by NSF CAREER award
1653648.}}
\date{}
\begin{document}
\maketitle
\thispagestyle{empty}
\pagestyle{empty}
\begin{abstract}
We consider the problem of distributed hypothesis testing (or social learning) where a network of agents seeks to identify the true state of the world from a finite set of hypotheses, based on a series of stochastic signals that each agent receives.  Prior work on this problem has provided distributed algorithms that guarantee asymptotic learning of the true state, with corresponding efforts to improve the rate of learning.  In this paper, we first argue that one can readily modify existing asymptotic learning algorithms to enable learning in finite time, effectively yielding arbitrarily large (asymptotic) rates.  We then provide a simple algorithm for finite-time learning which only requires the agents to exchange a binary vector (of length equal to the number of possible hypotheses) with their neighbors at each time-step.  Finally, we show that if the agents know the diameter of the network, our algorithm can be further modified to allow all agents to learn the true state and stop transmitting to their neighbors after a finite number of time-steps.
\end{abstract}
\section{Introduction}
A key challenge in the control of networked autonomous systems is to enable the entire system to accurately learn the state of the environment that it is operating in, despite the fact that measurements of that environment may be dispersed throughout the system. Typically, the information gathered by each member of the network provides only a partial view of the global system state, necessitating collaboration amongst members to learn the true state of the environment. 

One such scenario arises when the true state of the world is an element of a finite set $\Theta$ of possible states (or {\it hypotheses}), and each member of a network of agents  receives a stream of stochastic measurements of the environment (where the statistics of the measurements are a function of the true state).  Each agent is required to maintain a belief vector (i.e., a probability distribution) over the set of possible states, which it then updates based on its local measurements and information that it exchanges with neighbors.  This problem has been studied under various names in the literature, including {\it distributed hypothesis testing}, {\it distributed inference}, and {\it social learning}.  The key questions in this class of problems include: (i) What (or how much) information should the agents exchange with their neighbors at each time-step? (ii) How should the agents update their beliefs over the set of possible states based on their local measurements and the information from their neighbors? (iii) What is the fastest {\it rate} at which the true state can be learned in such settings?  

This class of problems has been studied for several decades, initially for scenarios involving a centralized fusion center \cite{fusion1}, and more recently in fully distributed settings where agents are interconnected over a network \cite{jad1,jad2,liu,rad,shahinparam,shahinTAC,nedic,nedic2,lalitha1,lalitha2,su1,mitraACC19}.
The distributed algorithms provided in these latter papers require each agent to iteratively combine belief vectors obtained from their neighbors with Bayesian updates involving their local signals \cite{jad1,jad2,liu,rad,shahinparam,shahinTAC,nedic,nedic2,lalitha1,lalitha2,su1,mitraACC19}.  
These rules ensure that all agents asymptotically learn the true state of the world, with the main differences being in the rate of learning.  Specifically, the linear and log-linear updating rules proposed in \cite{jad1, jad2,nedic,nedic2,lalitha1,lalitha2,shahinTAC} ensure that beliefs on false hypotheses exponentially decay to zero, at a rate that is determined by a convex combination of the relative entropies of the distributions of the signals received by the agents.  The paper \cite{mitraACC19} proposed a  different approach based on a ``min'' rule, which improves on these asymptotic rates: it ensures that beliefs on each false hypothesis decay exponentially fast at a rate given by the {\it largest} relative entropy between the true state and that false hypothesis over all agents.  
Recently, \cite{mitra2019communication}  showed that exponentially fast learning can be obtained even when the inter-communication intervals grow exponentially over time.


\subsection*{Contributions of this paper}  
First, we provide a simple result  showing that for any algorithm that enables learning asymptotically, there is a straightforward modification of that algorithm that enables learning in finite-time.  This implies that {\it arbitrarily fast} learning rates can be achieved to solve this problem.

Second, we provide a simple algorithm that not only provides finite-time learning, but {\it also}   only requires the agents to exchange a {\it binary vector} (of size equal to the number of hypotheses) at each iteration, as opposed to exchanging probability distributions as in all existing works. 

Third, we show that if each agent knows the diameter of the network, our algorithm can be modified to ensure that all agents learn the true state in finite time, exchange only an $m$-bit vector with their neighbors at each time-step,  {\it and stop communicating} with their neighbors after a finite number of time-steps, almost surely.

Our algorithms make the same assumptions that essentially all of the previously discussed works require (other than  knowledge of the diameter for the third contribution above),  
and significantly reduce the amount of communication required for learning compared to existing approaches, where agents have to communicate infinitely often. 

\subsection*{Notation and Terminology}
We will use the notation $\mathcal{G}=\{\mathcal{V},\mathcal{E}\}$ to denote a {\it graph} (or network), where $\mathcal{V} = \{1, 2, \ldots, n\}$ is a set of nodes (or agents), and $\mathcal{E} \subseteq \mathcal{V}\times\mathcal{V}$ is a set of edges.  A sequence of nodes $i_1, i_2, \ldots, i_k$ is said to be a {\it path} if $(i_j, i_{j+1}) \in \mathcal{E}$ for all $j \in \{1, 2, \ldots, k-1\}$; the {\it length} of the path is equal to $k-1$. Given two nodes $i, j \in \mathcal{V}$, the {\it distance} from $i$ to $j$ in $\mathcal{G}$ is the length of the shortest path from $i$ to $j$, and denoted by $d(i,j)$.  The graph is said to be {\it strongly connected} if, for all pairs of nodes $i, j \in \mathcal{V}$, there is a path from $i$ to $j$. The {\it diameter} of the graph is the maximum distance over all pairs of nodes, and is denoted by $D(\mathcal{G}) = \max_{i,j \in \mathcal{V}}d(i,j)$.
For each node $i \in \mathcal{V}$, the set of {\it in-neighbors} is denoted by $\mathcal{N}_i^{-} = \{j \in \mathcal{V} \vert (j,i) \in \mathcal{E}\}$, and the set of {\it out-neighbors} is denoted by $\mathcal{N}_i^{+} = \{j \in \mathcal{V} \vert (i,j) \in \mathcal{E}\}$.

We use $\mathbf{1}_m$ to denote the column vector of length $m$ with all elements equal to $1$.  Given a set of binary vectors $v_1, v_2, \ldots, v_k \in \{0,1\}^m$, we take the {\it intersection} of those vectors to be a binary vector $v \in \{0,1\}^m$, with the property that for all $j \in \{1, 2, \ldots, m\}$, the $j$-th element of $v$ is equal to `$1$' if and only if the $j$-th element of all of the vectors $v_1, v_2, \ldots, v_k$ is equal to `$1$'.  We denote this operation by \texttt{intersect}$(\{v_1, v_2, \ldots, v_k\})$, and note that it can be performed by simply taking an element-wise minimum or product of the given binary vectors.

\section{Problem Formulation}
Consider a network of agents modeled by a time-invariant and strongly connected graph $\mathcal{G} = \{\mathcal{V},\mathcal{E}\}$, where $\mathcal{V} = \{1, 2, \ldots, n\}$ is the set of $n$ agents, and $\mathcal{E} \subset \mathcal{V}\times \mathcal{V}$ is the set of edges, indicating communication capabilities between the agents.  In particular, the presence of an edge $(i,j) \in \mathcal{E}$ indicates that agent $i$ can transmit information to agent $j$.  

The network of agents is tasked with determining the true state of the world from a set of possible hypotheses $\Theta = \{\theta_1, \theta_2, \ldots, \theta_m\}$.  We will denote the true state (unknown to the agents a priori) as $\theta^* \in \Theta$.  Each agent has a sensor that receives a stream of stochastic measurements, whose statistics are dependent on the underlying state of the world.  More specifically, at each time-step $t \in \mathbb{N}$, each agent $i \in \mathcal{V}$ receives a measurement $s_{i,t} \in \mathcal{S}_i$, where $\mathcal{S}_i$ denotes the (finite) signal space of agent $i$.  For each possible realized state $\theta \in \Theta$, the measurement $s_{i,t}$ is a random variable whose distribution is denoted by $l_i(\cdot \vert \theta)$.  In particular, for each $w \in \mathcal{S}_i$, the quantity $l_i(w | \theta)$ denotes the probability that the measurement $s_{i,t}$ takes the value $w$ at time-step $t$ when the true state of the world is $\theta \in \Theta$.  We make the following standard assumptions  \cite{jad1,jad2,shahinTAC,nedic,lalitha1}.

\begin{assumption}
\label{assump:states_and_signals}
The signals and states satisfy the following properties:
\begin{itemize}
    \item For all states $\theta \in \Theta$ and for all agents $i \in \mathcal{V}$, the measurements seen by agent $i$ are independent and identically distributed over time.
    \item Each agent $i \in \mathcal{V}$ knows the set of  distributions $\{l_i(\cdot|\theta) : \theta \in \Theta\}$ of the measurements it would see under each possible state of the world (but does not know the distributions of the measurements received by other agents).
    \item For all states $\theta \in \Theta$,
    agents $i \in \mathcal{V}$, and measurement values $w \in \mathcal{S}_i$, the distributions satisfy $l_i(w \vert \theta) > 0$.
    \item There exists a single (and fixed) true state $\theta^* \in \Theta$ that is unknown to all the agents, and which generates the measurements seen by all the agents.
    \end{itemize}
\end{assumption}
As we will argue later, some of these assumptions can be relaxed for our algorithm.  We also make the following assumption purely for ease of exposition (our algorithm can be easily extended for situations with general priors).

\begin{assumption}
Each agent $i \in \mathcal{V}$ starts with a uniform prior on each of the possible states, denoted by the vector $\pi_{i,0} = \frac{1}{m}\mathbf{1}_m$.
\label{assump:uniform_priors}
\end{assumption}

Let $s_t = \left[\begin{matrix}s_{1,t} & s_{2,t} & \cdots & s_{n,t}\end{matrix}\right]'$ denote the vector of measurements seen by all agents at time-step $t$, and denote $\mathcal{S} = \mathcal{S}_1 \times \mathcal{S}_2 \times \cdots \times \mathcal{S}_n$, so that $s_t \in \mathcal{S}$ for all $t \in \mathbb{N}$.  We denote the distribution of $s_t$ under a given state $\theta \in \Theta$ by $l(\cdot|\theta)$.  We define a probability space for the stream of measurement vectors by $(\Omega, \mathcal{F}, \mathbb{P}^{\theta^*})$, where  $\Omega\triangleq\{\omega: \omega=(s_1,s_2,\ldots), \forall s_t\in\mathcal{S}, \forall t \in \mathbb{N}_{+}\}$, $\mathcal{F}$ is the $\sigma$-algebra generated by the observation profiles, and $\mathbb{P}^{\theta^{\star}}$ is the probability measure induced by sample paths in $\Omega$. Specifically, $\mathbb{P}^{\theta^{\star}}=\prod \limits_{t=1}^{\infty}l(\cdot|\theta^{\star})$. We will use the abbreviation a.s. to indicate almost sure occurrence of an event w.r.t. $\mathbb{P}^{\theta^{\star}}$. 

\subsection*{Objective}  
The goal of all agents in the network is to learn the true state $\theta^*$.  However, no individual agent's measurements may be sufficiently informative to allow it to learn $\theta^*$ on its own.  Thus, the agents have to exchange information with their neighbors in the network, and update their beliefs over the set of states in such a way that all agents eventually learn the true state.  More specifically, each agent $i \in \mathcal{V}$ maintains a belief vector (i.e., a probability distribution over $\Theta$) $\mu_{i,t} \in [0,1]^m$, which it updates over time based on its received measurements and information from its neighbors.  We will denote the element of $\mu_{i,t}$ corresponding to a particular state $\theta \in \Theta$ by $\mu_{i,t}(\theta)$. We will also use $e_{\theta^*} \in \{0,1\}^m$ to denote the indicator vector with a single `1' in the entry corresponding to $\theta^*$, and zeros everywhere else.  The distributed hypothesis testing problem is defined as follows.

\begin{problem}\label{prob:dist_ht}
Design a set of information exchange and belief update rules so that for all agents $i\in \mathcal{V}$, the belief vector $\mu_{i,t}$ converges to the vector $e_{\theta^*}$ (a.s.), i.e., for all $i\in \mathcal{V}$, $\mu_{i,t}(\theta^*) \rightarrow 1$ as $t \rightarrow \infty$ a.s.
\end{problem}

As we noted in the introduction, there are a variety of algorithms that have been proposed to solve this problem (asymptotically) \cite{jad1,jad2,shahinTAC,nedic,lalitha1,mitraACC19}.  
We will first show that these algorithms can be modified in a straightforward manner to obtain finite-time learning (a.s.).  In other words, for each sample path in a set of measure $1$, these algorithms can be modified so that there exists a finite (sample path dependent) $T \in \mathbb{N}$ such that for all $i \in \mathcal{V}$, $\mu_{i,t}(\theta^*) = 1$ and $\mu_{i,t}(\theta) = 0$ for all $t \ge T$ and all $\theta \in \Theta \setminus \{\theta^*\}$.
We then develop a simple learning rule that provides finite-time learning while only requiring the agents to exchange binary vectors for a finite number of time-steps.


\section{A General Result on Finite-Time Distributed Hypothesis Testing}
\label{sec:finite_time_learning}
We start with the following  result, showing that a large class of existing algorithms that provide asymptotic learning can be easily modified to provide finite-time convergence.

\begin{proposition}
Consider an algorithm $\mathcal{A}$ that solves Problem~\ref{prob:dist_ht}, and let $\mu_{i,t}$, $i \in \mathcal{V}$ be the belief vectors maintained by each agent under that algorithm.  For all $i \in \mathcal{V}$, let agent $i$ run the algorithm $\mathcal{A}$, but also maintain an additional vector $\bar{\mu}_{i,t} \in \{0,1\}^m$ at each time-step of the algorithm, where $\bar{\mu}_{i,t}$ has a `1' in the element where $\mu_{i,t}$ has its largest value (breaking ties arbitrarily), and zeros everywhere else.  Then $\bar{\mu}_{i,t}$ converges to $e_{\theta^*}$ in finite time under algorithm $\mathcal{A}$ a.s.
\label{prop:finite_time_learning}
\end{proposition}

\begin{proof}
Under any algorithm $\mathcal{A}$ that solves Problem~\ref{prob:dist_ht}, let $\bar{\Omega} \subset \Omega$ be the set of sample paths (of measure 1) for which the beliefs $\mu_{i,t}$ held by each agent converge to the vector $e_{\theta^*}$.  For each $\omega \in \bar{\Omega}$, we have  $\mu_{i,t}(\theta^*) \rightarrow 1$ and $\mu_{i,t}(\theta) \rightarrow 0$ for all $\theta \in \Theta \setminus \{\theta^*\}$ along that sample path.  Thus, for all $\omega \in \bar{\Omega}$ there exists $T(\omega) \in \mathbb{N}$ such that for all $i \in \mathcal{V}$, for all $t \ge T(\omega)$, and for all $\theta \ne \theta^*$, we have $\mu_{i,t}(\theta^*) > \mu_{i,t}(\theta)$.  Thus, for all $\omega\in \bar{\Omega}$, the vector $\bar{\mu}_{i,t}$ specified in the proposition will take on the value $e_{\theta^*}$ for all $t \ge T(\omega)$ and all $i \in \mathcal{V}$.
\end{proof}

The above result, which is perhaps obvious in hindsight, does provide some insights into the problem of distributed hypothesis testing.  In particular, it suggests that {\it arbitrarily fast rates} of learning can be achieved for this problem, simply by modifying existing algorithms in a straightforward manner. It is worth noting that the above result {\it does not} detract from existing asymptotic learning algorithms.  On the contrary, the above result introduces new metrics (other than asymptotic rates of learning) that would be of interest to understand in the context of those algorithms.  For example, which asymptotic learning algorithm, when modified as in Proposition~\ref{prop:finite_time_learning}, would yield the smallest time (in an appropriate probabilistic sense) to learn the true state?  Is there a relationship between the asymptotic rate of learning (for the base algorithm) and the finite time guarantee provided by the amended algorithm?  These are just some of the questions that would be worth pursuing for future research.

Here, we turn our attention to another question, namely understanding how much information the agents need to exchange with each other to solve the distributed hypothesis testing problem in finite-time.

\section{Towards a Communication-Efficient Finite-Time Algorithm: Gaining Intuition}
We start by establishing some preliminary concepts that will provide intuition for our eventual algorithm. 

\subsection{What Can Each Agent Do with Local Bayesian Updates?}
First, as in \cite{mitraACC19}, we ask the question, ``What information can each agent infer about the states based purely on its own signals?''  To answer this question, we will need the concept of {\it distinguishability} between a pair of states $\theta_p, \theta_q \in \Theta$.  

\begin{definition}{(\textbf{Distinguishable States})}\label{def:distinguishable_states}
Consider a distinct pair of states $\theta_p, \theta_q \in \Theta$.  We say that these states are {\it distinguishable} by agent $i \in \mathcal{V}$ if $D(l_i(\cdot|\theta_p)||l_i(\cdot|\theta_q)) > 0$, where $D(l_i(\cdot|\theta_p)||l_i(\cdot|\theta_q))$ represents the KL-divergence \cite{cover} between the distributions $l_i(\cdot|\theta_p)$ and $l_i(\cdot|\theta_q)$.  On the other hand, if $D(l_i(\cdot|\theta_p)||l_i(\cdot|\theta_q)) = 0$, we say that states $\theta_p$ and $\theta_q$ are {\it indistinguishable} by agent $i$.
\end{definition}

Note that the KL-divergence between two distributions $l_i(\cdot|\theta_p)$ and $l_i(\cdot|\theta_q)$ is always nonnegative, and zero if and only if the two distributions are exactly the same (over the finite signal space $\mathcal{S}_i$).  Thus, distinguishability between $\theta_p$ and $\theta_q$ by agent $i$ implies that the signals seen by agent $i$ under each of those different states will have different statistics.  Based on this, we will define the following sets.

\begin{definition}\label{def:possible_true_states}
For each agent $i \in \mathcal{V}$, and for each state $\theta \in \Theta$, define the set $\Theta_i^{\theta} \subseteq \Theta$ to be the set of all states that are indistinguishable from the  state $\theta$ by agent $i$.  In particular, the set $\Theta_i^{\theta^*} \subseteq \Theta$ is the set of states that are indistinguishable from the true state by agent $i$.
\end{definition}

In other words, the signals seen by agent $i$ under the true state $\theta^*$ will never allow it to distinguish between the states in $\Theta_i^{\theta^*}$.  To make this more precise, we  now discuss how distinguishability between two states can  be leveraged by the agent to determine which (if any) of those two states could possibly be true based on its local signals.  Consider a simple Bayesian update performed by agent $i$, of the form
\begin{equation}
\pi_{i,t+1}(\theta)=\frac{l_i(s_{i,t}|\theta)\pi_{i,t}(\theta)}{\sum  \limits_{p=1}^{m} l_i(s_{i,t}|\theta_p)\pi_{i,t}(\theta_p)}, \enspace \forall \theta \in \Theta,
\label{eqn:Bayes}
\end{equation}
where $\pi_{i,0}(\theta) = \frac{1}{m}$, $\forall \theta \in \Theta$, is a uniform prior on all states.  Following the terminology in \cite{mitraACC19}, we will refer to $\pi_{i,t}(\theta)$ as the {\it local belief} of agent $i \in \mathcal{V}$ on the state $\theta \in \Theta$ at the start of time-step $t$.  
 
The following result shows the behavior of the local beliefs generated by the (local) Bayesian update \eqref{eqn:Bayes}.

\begin{lemma}(\cite{mitraTAC19})
Consider the local Bayesian update given by \eqref{eqn:Bayes}.  The following properties hold.
\begin{itemize}
    \item For all states $\theta \in \Theta \setminus \Theta_i^{\theta^*}$, $\pi_{i,t}(\theta) \rightarrow 0$ a.s. as $t \rightarrow \infty$.
    \item For all $\theta \in \Theta_i^{\theta^*}$,  $\pi_{i,t}(\theta) \rightarrow \frac{1}{|\Theta_i^{\theta^*}|}$ a.s. as $t \rightarrow \infty$.
\end{itemize}
\label{lemma:Bayes}
\end{lemma}

\textbf{Key Insight.} The above result shows that 
the local beliefs maintained by each agent $i$ will {\it separate} into two levels, with beliefs on states in $\Theta_i^{\theta^*}$ going to $\frac{1}{|\Theta_i^{\theta^*}|}$, and beliefs on states in $\Theta \setminus \Theta_i^{\theta^*}$ going to zero a.s. 
In particular, since $|\Theta_i^{\theta^*}| \le |\Theta| = m$, we have the following useful corollary of Lemma~\ref{lemma:Bayes}.

\begin{corollary}\label{cor:finite_time_local_beliefs}
Consider the probability space $(\Omega, \mathcal{F}, \mathbb{P}^{\theta^*})$ for the signals seen by the agents, and suppose each agent runs the Bayesian update rule \eqref{eqn:Bayes} to update its local beliefs on the set of states. Suppose Assumption~\ref{assump:states_and_signals} and Assumption~\ref{assump:uniform_priors} hold.  Then, there is a set $\bar{\Omega}\subset \Omega$ of measure $1$ with the following property.  For all $\omega \in \bar{\Omega}$ and for all $\alpha \in (0,1)$, there exists a $T(\omega) \in \mathbb{N}$ such that for all $i \in \mathcal{V}$:
\begin{itemize}
    \item For all $\theta \in \Theta_i^{\theta^*}$, $\pi_{i,t}(\theta) > \frac{\alpha}{m}$ for all $t \ge T(\omega)$.
    \item For all $\theta \in \Theta \setminus \Theta_i^{\theta^*}$, $\pi_{i,t}(\theta) \le \frac{\alpha}{m}$ for all $t \ge T(\omega)$.
     \end{itemize}
\end{corollary}

The parameter $\alpha \in (0,1)$ can be arbitrary (and this is the reason we elide the dependence of $T(\omega)$ on $\alpha$).  The above result shows that for any fixed $\alpha \in (0,1)$, along each sample path in a set of measure $1$, there is a finite time for each agent after which its local signals are {\it no longer helpful} for it to identify the true state.  In particular, along sample path $\omega$ and for some fixed $\alpha \in (0,1)$, {\it suppose} agent $i$ knew the time $T(\omega)$;\footnote{Of course, it is not apparent {\it how} the agent would be able to identify this time. We will show how to circumvent this when we present our algorithm in the next section, but we continue our thought experiment for now.} then, at this time, agent $i$ can  identify the set $\Theta_i^{\theta^*}$ by simply checking which states have beliefs larger than $\frac{\alpha}{m}$. Note that the agent still would not know {\it which} of the states within $\Theta_i^{\theta^*}$ is the true $\theta^*$.  In the next subsection, we discuss how the agents can resolve this ambiguity by exchanging information over the network.

\subsection{How Should the Network Collectively Leverage the Local Knowledge at Each Agent to Learn the True State?}

Consider a sample path $\omega$ from the set $\bar{\Omega}$ of sample paths of measure $1$ identified by Corollary~\ref{cor:finite_time_local_beliefs}.  Suppose that along that sample path, each agent $i \in \mathcal{V}$ has identified the set $\Theta_i^{\theta^*}$ at time $T(\omega)$, as discussed in the previous subsection.  How should the agents work together to determine the one true state from their individual knowledge of these sets?  To answer this, suppose that the following assumption holds.

\begin{assumption}{(\textbf{Global Identifiability})}\label{assump:global_identifiability}
For all pairs of distinct states $\theta_p, \theta_q \in \Theta$, there exists at least one agent $i \in \mathcal{V}$ for which $\theta_p$ and $\theta_q$ are distinguishable by that agent.
\end{assumption}

The above assumption is made in almost all of the existing literature on distributed hypothesis testing \cite{jad1, jad2, shahinTAC, nedic,lalitha1,mitraACC19}.  A simple implication of the above assumption is that $\cap_{i \in \mathcal{V}}\Theta_i^{\theta^*} = \theta^*$.  

\textbf{Key Insight.} Once each agent $i \in \mathcal{V}$ determines the set $\Theta_i^{\theta^*}$, if the agents simply find the intersection of those sets (i.e., use the {\it process of elimination}), they will  identify the true state $\theta^*$ (under the global identifiability condition).  

At this point, the following facts should be clear to the reader: under Assumption \ref{assump:states_and_signals}, Assumption~\ref{assump:uniform_priors}, and Assumption~\ref{assump:global_identifiability}, (i) there exists a finite time (a.s.) after which each agent $i$'s local beliefs will allow it to recover the set $\Theta_i^{\theta^*}$, and (ii) the agents can identify the true state $\theta^*$ by finding the intersection of those sets. The question now is how to account for the fact that each agent $i$ will not be able to identify the time at which it can conclusively determine $\Theta_i^{\theta^*}$.  We now develop a simple algorithm that circumvents this issue.

\section{A Communication-Efficient Algorithm for Finite-Time Distributed Hypothesis Testing}

We present Algorithm~\ref{algo:PoE}, which we call the Process of Elimination (PoE) algorithm.  Below, we walk through the steps and components of the algorithm.

\begin{algorithm}
\caption{\textbf{(PoE algorithm)}  Each agent $i \in \mathcal{V}$ executes this algorithm in parallel} \label{algo:PoE}
\textbf{Input} A set of time-indices $\mathcal{I}$, denoting epochs
\begin{algorithmic}[1]
\State Agent $i$ initializes $\pi_{i,0} = \frac{1}{m}\mathbf{1}_m$,  $\mu_{i,0} = \pi_{i,0}$, $\psi_{i,0} = \mathbf{1}_m$
\For {$t \in \mathbb{N}$}  
\If {$t \in \mathcal{I}$}
\State Set $\psi_{i,t} \gets$ \texttt{round}$(\pi_{i,t})$
\EndIf
\State Send $\psi_{i,t}$ to all out-neighbors
\State Receive $\psi_{j,t}$ from all in-neighbors $j \in \mathcal{N}_i^{-}$
\State $\psi_{i,t+1} \gets$ \texttt{intersect}$\left(\left\{\psi_{j,t}, j \in \mathcal{N}_i^{-} \cup \{i\}\right\}\right)$
\If {$t+1 \in \mathcal{I}$}
\State Set $\mu_{i,t+1} \gets \frac{1}{\|\psi_{i,t+1}\|_1}\psi_{i,t+1}$ 
\Else
\State Set $\mu_{i,t+1} \gets \mu_{i,t}$
\EndIf
\State Update $\pi_{i,t+1}$ based on the local Bayesian rule \eqref{eqn:Bayes}
\EndFor
\end{algorithmic}
\end{algorithm}

\subsection{Components of the PoE Algorithm}

\subsubsection*{Epochs}
We partition the (discrete) time axis into a set of nonoverlapping contiguous intervals.  Specifically, we define an infinite set of time-indices $\mathcal{I} \subseteq \mathbb{N}$, with $\mathcal{I} = \{t_0, t_1, t_2, \ldots\}$.  We take $t_0 = 0$ without loss of generality. For $k \in \mathbb{N}$, we denote the interval $\{t_k, t_k+1, \ldots, t_{k+1}-1\}$ by $\mathcal{B}_k$, and refer to it as {\it epoch} $k$.  Thus, each time-index in $\mathcal{I}$ indicates the first time-step of an epoch.   We make the following assumptions on the epochs.

\begin{assumption}\label{assump:nondecreasing_epochs}
The epochs are nondecreasing in length, i.e., $|\mathcal{B}_{k+1}| \ge |\mathcal{B}_k|$ for all $k \in \mathbb{N}$.
\end{assumption}

Essentially, at the start of each epoch, each agent $i \in \mathcal{V}$ will form an estimate of the set $\Theta_i^{\theta^*}$, based on its local beliefs (leveraging Corollary~\ref{cor:finite_time_local_beliefs}).  During the rest of the epoch, the agents will find the intersection of those sets, motivated by the discussion in the previous section.  At the end of each epoch, each agent will attempt to identify the true state based on the intersection of the local sets. The agents will repeat this process in each epoch.

\subsubsection*{Vectors maintained by each agent}

To enable the process described above, the PoE algorithm requires each agent $i \in \mathcal{V}$ to maintain three vectors:
\begin{itemize}
    \item The vector $\pi_{i,t} \in [0,1]^m$ represents the local belief vector maintained by each agent.  This vector is initialized with the uniform distribution $\pi_{i,0} = \frac{1}{m}\mathbf{1}_m$ in Line 1 of the algorithm, and is updated according to Bayes' rule (equation \eqref{eqn:Bayes}) at each time-step (Line 14).
    \item The vector $\mu_{i,t} \in [0,1]^m$ represents agent $i$'s beliefs on the set of states, {\it incorporating} the information received from neighbors.  We refer to $\mu_{i,t}$ as the {\it network belief vector} maintained by agent $i$.  This vector is initialized as $\mu_{i,0} = \pi_{i,0}$ in Line 1 of the algorithm.
    \item The binary vector $\psi_{i,t} \in \{0,1\}^m$ is maintained and updated by each agent at each time-step of the algorithm, and is used to calculate the intersection of the sets of potential true hypotheses calculated by each agent.
\end{itemize}

\subsubsection*{Rounding function}
At the start of each epoch $\mathcal{B}_k$, $k \in \mathbb{N}$, we require each agent $i \in \mathcal{V}$ to form an estimate of the set $\Theta_i^{\theta^*}$.  Based on Corollary~\ref{cor:finite_time_local_beliefs}, we do this as follows.  Using the local belief vector $\pi_{i,t_k}$ (which is the belief at the first time-step $t_k$ of epoch $\mathcal{B}_k$), we define the function \texttt{round}$(\pi_{i,t_k})$ to return a binary vector of size $m$.  Specifically, for $j \in \{1, 2, \ldots, m\}$, the $j$-th entry of the returned vector is equal to $1$ if $\pi_{i,t_k}(\theta_j) > \frac{\alpha}{m}$ (for some fixed, but arbitrary, parameter $\alpha \in (0,1))$,\footnote{Our algorithm will guarantee finite-time learning for any $\alpha \in (0,1)$; however, the choice of $\alpha$ will affect the transient behavior of the algorithm.  If $\alpha$ is close to $1$, then the true state may be eliminated for some period of time if some agent's signals cause it to place a low belief on that state. On the other hand, if $\alpha$ is set close to $0$, it will take longer for the beliefs on the false states to fall below the threshold $\frac{\alpha}{m}$, which means that it will take longer for each agent to accurately identify its set $\Theta_i^{\theta^*}$.  Nevertheless, we find that setting $\alpha$ to be small works well in practice.} and is equal to zero otherwise.  This rounding step is done in Lines 3-5 of the algorithm.

\subsubsection*{Distributed set intersection}
At each time-step of each epoch, the agents seek to find the intersection of their local sets of potential true hypotheses.  Specifically, in each epoch $\mathcal{B}_k$ (starting at time-step $t_k \in \mathcal{I}$), recall that $\psi_{i,t_k}$ is set in Line 4 to be the binary vector indicating agent $i$'s estimate of the set $\Theta_i^{\theta^*}$.  At each time-step $t \in \mathcal{B}_k$ of the epoch, each agent transmits its current vector $\psi_{i,t}$ to its out-neighbors (Line 6), and receives the vectors $\psi_{j,t}$ of each in-neighbor $j \in \mathcal{N}_i^{-}$ (Line 7).  Based on these received vectors, agent $i$ finds the intersection of the sets indicated by those vectors in Line 8.  Note that if each agent in a network starts with a binary vector, and each agent iteratively updates its vector by intersecting it with the vectors of its neighbors as above, then after $D$ time-steps (where $D$ is the diameter of the network), the vector maintained by all agents will be the intersection of all initial vectors in the network. We will use this fact in the analysis of the PoE algorithm.

\subsubsection*{Updating the network belief vector}

The network belief vector $\mu_{i,t}$ maintained by each agent $i \in \mathcal{V}$ is updated only at the last time-step of each epoch (captured by the test in Line 9 of the algorithm); it is held constant for all other time-steps of each epoch.  Specifically, at the last time-step $t$ of each epoch, each agent $i \in \mathcal{V}$ calculates its network belief vector $\mu_{i,t+1}$ based on the set intersection vector $\psi_{i,t+1}$ it has computed at that time-step, as indicated by Line 10 of the algorithm.  In particular, Line 10 takes the binary vector $\psi_{i,t+1}$ and normalizes it to be a probability distribution over the set of states.\footnote{In case $\psi_{i,t+1}$ is the zero vector, we interpret $\frac{1}{\|\psi_{i,t+1}\|_1}\psi_{i,t+1}$ as the vector $\frac{1}{m}\mathbf{1}_m$, i.e., equal beliefs on all states.} In the ideal case, the vector $\psi_{i,t+1}$ will have a single $1$ on the true state $\theta^*$, and zeros elsewhere, in which case the network belief vector $\mu_{i,t+1}$ will also have a single $1$ on the true state, and zeros elsewhere; we will show that this will indeed happen after a finite number of time-steps a.s., under the assumptions that we are considering.

\subsection{Analysis of the PoE Algorithm}
We now prove the following key result.

\begin{theorem}\label{thm:poe_result}
Let $\mathcal{G} = \{\mathcal{V},\mathcal{E}\}$ be the network of agents, and let $\mathcal{I} \subseteq \mathbb{N}$ be an infinite set indicating the starting time-steps of the epochs. Suppose Assumption~\ref{assump:states_and_signals}, Assumption~\ref{assump:uniform_priors},  Assumption~\ref{assump:global_identifiability}, and Assumption~\ref{assump:nondecreasing_epochs} hold.  If at least one epoch has length larger than the diameter of the network, then the PoE algorithm guarantees that for all $i\in\mathcal{V}$, the network belief $\mu_{i,t}$ converges to $e_{\theta^*}$ in finite time almost surely.
\end{theorem}

\begin{proof}
Under Assumption 1, let $\bar{\Omega} \subset \Omega$ be the set of sample paths of measure $1$ indicated by Corollary $1$.  Based on that corollary, for each $\omega \in \bar{\Omega}$, let $T(\omega)$ be the finite time after which the local belief vector of each agent $i\in \mathcal{V}$ has separated, with the beliefs on states in the set $\Theta_i^{\theta^*}$ being larger than $\frac{\alpha}{m}$, and the beliefs on states in the set $\Theta \setminus \Theta_i^{\theta^*}$ being less than or equal to $\frac{\alpha}{m}$.  Fix an $\omega \in \bar{\Omega}$ for the rest of the proof; the same arguments will hold for all $\omega \in \bar{\Omega}$.

Based on Assumption~\ref{assump:nondecreasing_epochs} (nondecreasing epoch lengths) and the condition in the statement of the theorem that at least one epoch has length larger than the diameter of the network, we know that all epoch lengths will eventually become larger than the diameter of the network.  Let $K \in \mathbb{N}$ be index of the first epoch that has length larger than the diameter of the network {\it and} that starts after $T(\omega)$ (i.e., $t_{K} \ge T(\omega)$).  

Consider time-step $t_K$ in Algorithm~\ref{algo:PoE}.  Since $t_K \in \mathcal{I}$, every agent $i \in \mathcal{V}$ will set $\psi_{i,t_K} \gets$ \texttt{round}$(\pi_{i,t_K})$ in line 4 of the algorithm.  Since $t_K \ge T(\omega)$, and based on the definition of the rounding function, we see that $\psi_{i,t_K}$ will be a binary vector with a $1$ on every element in the set $\Theta_i^{\theta^*}$, and a zero on every other element.  In other words, the vector $\psi_{i,t_K}$ will exactly represent $\Theta_i^{\theta^*}$ for every agent $i \in \mathcal{V}$.  

Now consider the remaining time-steps in epoch $K$.  Based on Algorithm~\ref{algo:PoE}, for all $t \in \mathcal{B}_K$, the vector $\psi_{i,t}$ is updated by every agent by intersecting its current vector with those of its in-neighbors.  Since the length of epoch $K$ is larger than the diameter of the network, 
it is easy to see that at the end of the last time-step $t_{K+1}-1$ of the epoch, the vector $\psi_{i,t_{K+1}}$ will contain the intersection of all the vectors $\{\psi_{j,t_K}, j \in \mathcal{V}\}$.  Since each $\psi_{j,t_K}$ was an indicator vector for $\Theta_j^{\theta^*}$, and using Assumption~\ref{assump:global_identifiability} (Global Identifiability), we see that $\psi_{i,t_{K+1}}$ will contain a single $1$ in the location corresponding to $\theta^*$, and zeros everywhere else.  

Finally at the end of time-step $t_{K+1}-1$, each agent $i \in \mathcal{V}$ will update $\mu_{i,t_{K+1}}$ based on line 10 of Algorithm~\ref{algo:PoE}; this will result in $\mu_{i,t_{K+1}}$ having a single $1$ in the entry corresponding to $\theta^*$, and zeros everywhere else, for all $i \in \mathcal{V}$.  

The above analysis applies to every epoch $k$ with index larger than $K$, since each such epoch will be larger than the diameter of the network, and the quantity $\psi_{i,t_k}$ computed in Line 4 of the algorithm by each agent $i$ will exactly correspond to the set $\Theta_i^{\theta^*}$.  Since $\mu_{i,t}$ is only updated at the end of each epoch, we see that for all $t \ge t_{K+1}$, we have $\mu_{i,t} = \mu_{i,t_{K+1}}$ for all $i \in \mathcal{V}$.  Thus the network beliefs of all agents converge to $e_{\theta^*}$ in a finite number of time-steps.
\end{proof}

\subsection{Discussion}
It is worth emphasizing here that Algorithm~\ref{algo:PoE} circumvents the need for each agent to know the time $T(\omega)$ for each sample path $\omega \in \bar{\Omega}$ (identified in Corollary~\ref{cor:finite_time_local_beliefs}), and also the diameter of the network.  Regarding the time $T(\omega)$, since the algorithm has each agent ``reset'' the vector $\psi_{i,t}$ at the start of each epoch, that vector is guaranteed to be reset to the desired vector (capturing membership in the set $\Theta_i^{\theta^*}$) at some point in time (namely the start of the first epoch after $T(\omega)$).  

Second, by choosing the epoch lengths to be increasing over time, Algorithm~\ref{algo:PoE} removes the need for each agent to know the diameter of the network (recall that the distributed set intersection steps need to be iterated for a number of time steps at least equal to the diameter of the network in order to allow all agents to compute the intersection of all the local sets).  For example, if we choose the epoch start times in such a way that $|\mathcal{B}_k| = k+1$ for all $k \in \mathbb{N}$ (i.e., the epoch lengths increase linearly), then the lengths are guaranteed to eventually become larger than the diameter of the network.  If, however, each agent does know the diameter of the network, they can simply choose the epochs to be such that  $|\mathcal{B}_k|$ is equal to the diameter for all $k$.  

Note also that the only information exchanged at each time-step $t$ by the agents in Algorithm~\ref{algo:PoE} is their binary vector $\psi_{i,t}$.  Thus, this algorithm only requires each agent to transmit $m$ bits of information to its neighbors at each time-step, which can be significantly smaller than the number of bits required to encode and transmit probability distributions (as in existing distributed hypothesis testing algorithms).  

While the number of bits exchanged at each time-step is small, Algorithm~\ref{algo:PoE}  still requires the agents to continue communicating for all time (as they continue resetting their vectors $\psi_{i,t}$ at the start of each epoch, and running the set intersection steps).  In the next section, we show that if all of the agents know the diameter of the network, one can modify Algorithm~\ref{algo:PoE} to obtain the same benefits (finite-time learning with at most $m$ bits of communication per time-step) with a {\it finite number of communications}.


\section{Modifying the PoE Algorithm to Require A Finite Number of Communications}

Recall that in the PoE Algorithm described in the previous section, for each sample path in a set of measure $1$, there will be some $K \in \mathbb{N}$ such that for all epochs with indices $k \ge K$, the quantity $\psi_{i,t_k}$ calculated by each agent will be the indicator vector for the set $\Theta_i^{\theta^*}$.  Thus, subsequent epochs (past epoch $K$) do not add additional useful information, since the agents will simply be intersecting the same sets in each of those epochs.  This suggests that if we can identify when the vectors $\psi_{i,t}$ have stopped changing, the agents do not need to transmit further.  We can do this as follows.  At the start of each epoch, each agent $i$ calculates its vector $\psi_{i,t}$ as usual.  However, before it transmits that vector, it compares it to the vector that it calculated at the start of the previous epoch.  If the vector has not changed from the previous epoch, the agent does not transmit and simply waits.  If it receives a transmission from a neighbor during the epoch, then some other agent in the network must have initiated transmissions (spurred by a change in that agent's local vector); thus, the waiting agent also starts transmitting and participating in the distributed set intersection operations (with its local vector $\psi_{i,t}$).  In this way, once all agents' vectors $\psi_{i,t}$ have settled down to their final values, no further transmissions will be initiated.  Note that the epoch lengths will need to be of length at least twice the diameter ($D$) of the network, as it will potentially take $D$ time-steps for an agent to realize that some other agent has initiated transmissions, and then another $D$ times-steps for the set intersection iterations to converge to their final value.  In particular, each agent will need to know the diameter of the network, so that the initial epoch lengths can be set to be twice the diameter.

The modified algorithm is referred to as ``PoE-FC'' (Process of Elimination with Finite Communications), and shown in Algorithm~\ref{algo:PoE-FC}.  The algorithm introduces two new variables to the baseline PoE algorithm: a binary flag called `transmit', and a vector $\psi_{i,prev} \in \{0,1\}^m$.  At the first time-step $t_k$ of each epoch, if the quantity $\psi_{i,t_k}$ is different from the quantity $\psi_{i,t_{k-1}}$ calculated at the start of the previous epoch, the transmit flag is set to `true' (Line 6 of the algorithm).  The vector $\psi_{i,prev}$ is used to enable this comparison, and stores the value of the vector $\psi_{i,t_{k-1}}$ calculated at the start of the previous epoch.  If the vector calculated at the start of this epoch is the same as the one at the start of the previous epoch, the transmit flag is set to `false' (Line 9).  In Lines 12-14, agent $i$ only transmits its current vector $\psi_{i,t}$ to its out-neighbors if its transmit flag is `true'.  If any in-neighbor transmits its vector to agent $i$, then agent $i$ sets its own transmit flag to `true' to begin participating in the set intersection protocol (Line 16).  The rest of the algorithm proceeds in the same way as the original PoE algorithm.

We now prove the following result.

\begin{theorem}\label{thm:poe_fc_result}
Let $\mathcal{G} = \{\mathcal{V},\mathcal{E}\}$ be the network of agents, and let $\mathcal{I} \subseteq \mathbb{N}$ be an infinite set indicating the starting time-steps of the epochs. Suppose Assumption~\ref{assump:states_and_signals}, Assumption~\ref{assump:uniform_priors},  Assumption~\ref{assump:global_identifiability}, and Assumption~\ref{assump:nondecreasing_epochs} hold.  If each epoch has length at least equal to twice the diameter of the network, then the PoE-FC algorithm guarantees that for all $i\in\mathcal{V}$, the network belief $\mu_{i,t}$ converges to the true belief vector in finite time almost surely.  Furthermore, all agents stop transmitting after a finite number of time-steps almost surely.
\end{theorem}

\begin{proof}
Following the proof of Theorem~\ref{thm:poe_result}, we first define the set $\bar{\Omega} \subset \Omega$ of sample paths of measure 1 indicated by Corollary $1$.  As argued in the proof of Theorem~\ref{thm:poe_result}, for each $\omega \in \bar{\Omega}$ there exists $K \in \mathbb{N}$ such that for all $k \ge K$ and for all agents $i \in \mathcal{V}$, the vector $\psi_{i,t_k}$ computed at the start of epoch $\mathcal{B}_k$ is the indicator vector for the set $\Theta_i^{\theta^*}$.  We fix an $\omega \in \bar{\Omega}$ for the rest of the proof in order to present the argument.

For the $\omega \in \bar{\Omega}$ under consideration, let $K$ be the index of the {\it first} epoch where the above property holds.  Then, in the first time-step $t_K$ of that epoch, the test in Line 3 of Algorithm~\ref{algo:PoE-FC} will pass, and thus each agent $i$ will compute $\psi_{i,t_K}$ in Line 4 of the algorithm.  Since we are considering the first epoch where all of these computed vectors accurately reflect the corresponding sets $\Theta_i^{\theta^*}$ for their agents, there will be at least one agent $l$ such that $\psi_{l,t_{K}}$ is different from its previously computed vector $\psi_{l,prev}$.  Thus, that agent will set its transmit flag to `true' in Line 6 of the algorithm.  Next, based on that updated flag, the agent will transmit its vector $\psi_{l,t_K}$ to its out-neighbors in Line 13 of the algorithm; furthermore, since the transmit flag does not get reset again until the start of the next epoch, the steps followed by agent $l$ in Algorithm~\ref{algo:PoE-FC} will be identical to the steps it would have followed in the original POE Algorithm (Algorithm~\ref{algo:PoE}) for the rest of the epoch.  

Now consider an out-neighbor $l_1$ of agent $l$ at time-step $t_K$.  That agent will receive the transmission from agent $l$, and thus will set its transmit flag to `true' in Line 16 of Algorithm~\ref{algo:PoE-FC}.  Once again, since the transmit flag does not get reset until the start of the next epoch, the steps followed by agent $l_1$ in Algorithm~\ref{algo:PoE-FC} will be identical to the steps it would follow in Algorithm~\ref{algo:PoE} for the rest of the epoch.  

Now, consider any agent $l_2$ that is an out-neighbor of some out-neighbor $l_1$ of $l$. Since agent $l_1$'s transmit flag was set to `true' at time-step $t_K$, that agent will transmit its state to all its out-neighbors at time-step $t_K+1$ (based on Line 13 of Algorithm~\ref{algo:PoE-FC}).  Thus, agent $l_2$'s transmit flag will be set to `true' in Line 16 of the algorithm, at which point it remains true for the remainder of the epoch.  Repeating this argument, we see that all agents in the network will have their transmit flags set to `true' by time-step $t_K+D$, where $D$ is the diameter of the network.  Note that each agent $i$ keeps its vector $\psi_{i,t}$ constant (in Line 22 of the algorithm) until its transmit flag gets set to `true'.  Furthermore, 
if each agent $i$ has a $1$ in the $j$-th position of its vector $\psi_{i,t_K}$, then the \texttt{intersect} function will preserve the $1$ in the $j$-th location of each agent's vector for all time-steps in that epoch. Additionally, if any agent has a $0$ in some element of its vector, it will never change that entry to a $1$ during that epoch. Thus, at time-step $t_K+D$, every agent $i$ will have a $1$ in its vector $\psi_{i,t_K+D}$ in the location corresponding to the state $\theta^*$ (and furthermore, that will be the only shared location where every agent has a 1 in its vector).  Starting at that time-step, since all agents will execute Lines 19-20 at each iteration of Algorithm~\ref{algo:PoE-FC} for at least another $D$ time-steps (since each epoch is of length at least twice the diameter), we see that at the end of the epoch, the vector $\psi_{i,t_{K+1}}$ computed in Line 20 will be the indicator vector for the true state.  Furthermore, at time-step $t_{K+1}-1$, the test in Line 24 of the algorithm will pass for every agent, and thus all agents will set $\mu_{i,t_{K+1}}$ in Line 25 to be the indicator vector for the true state.  

For all subsequent epochs, since $\psi_{i,t_{k}}$ (computed in Line 4) will be the same as the vector computed at the start of the previous epoch (for all agents), all agents will set their transmit flag to `false' in Line 9.  Thus, no agent will ever transmit, and all agents will simply propagate their current (correct) $\mu_{i,t}$ vector forward for all time (as indicated by Line 27).  Consequently, all agents stop communicating and learn the true state in finite time, almost surely.
\end{proof}

\begin{algorithm}
\caption{\textbf{(PoE-FC algorithm)}  Each agent $i \in \mathcal{V}$ executes this algorithm in parallel} \label{algo:PoE-FC}
\textbf{Input} A set of time-indices $\mathcal{I}$, denoting epochs
\begin{algorithmic}[1]
\State Agent $i$ initializes $\pi_{i,0} = \frac{1}{m}\mathbf{1}_m$, $\psi_{i,0} = \mathbf{1}$, $\mu_{i,0} = \pi_{i,0}$, transmit = `false', $\psi_{i,prev} = \psi_{i,0}$
\For {$t \in \mathbb{N}$}  
\If {$t \in \mathcal{I}$}
\State Set $\psi_{i,t} \gets$ \texttt{round}$(\pi_{i,t})$
\If {$\psi_{i,t} \ne \psi_{i,prev}$}
\State Set transmit = `true'
\State $\psi_{i,prev} \gets \psi_{i,t}$
\Else
\State Set transmit = `false'
\EndIf
\EndIf
\If {transmit == `true'}
\State Send $\psi_{i,t}$ to all out-neighbors
\EndIf
\If {any in-neighbor transmits}
\State Set transmit = `true'
\EndIf
\If {transmit == `true'}
\State Receive $\psi_{j,t}$ from all in-neighbors $j \in \mathcal{N}_i^{-}$
\State $\psi_{i,t+1} \gets$ \texttt{intersect}$\left(\left\{\psi_{j,t}, j \in \mathcal{N}_i^{-} \cup \{i\}\right\}\right)$
\Else
\State $\psi_{i,t+1} \gets \psi_{i,t}$
\EndIf
\If {($t+1 \in \mathcal{I}$) AND (transmit == `true')}
\State Set $\mu_{i,t+1} \gets \frac{1}{\|\psi_{i,t+1}\|_1}\psi_{i,t+1}$ 
\Else
\State Set $\mu_{i,t+1} \gets \mu_{i,t}$
\EndIf
\State Update $\pi_{i,t+1}$ based on the local Bayesian rule \eqref{eqn:Bayes}
\EndFor
\end{algorithmic}
\end{algorithm}

\begin{figure*}[t]
  \centering
    \includegraphics[scale=0.5]{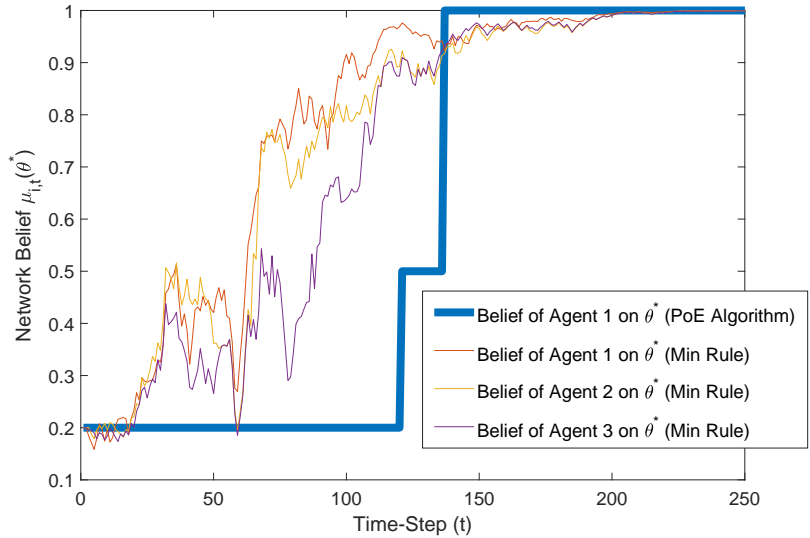}
  \caption{An illustration of the convergence of the network belief vector $\mu_{i,t}$ under the PoE Algorithm \ref{algo:PoE}, and for three agents in the network under the ``min''-rule algorithm from \cite{mitraACC19}.}
    \label{fig:network_belief_illustration}
\end{figure*}

\section{Simulation}
\label{sec:sim}
To illustrate the PoE algorithm, we generate a geometric random graph with 200 nodes, where each node is placed uniformly at random in the unit square.  We place an edge between two nodes if the Euclidean distance between them is at most 0.15, yielding a graph with a diameter of 11.  

We consider a hypothesis testing problem with a set $\Theta = \{\theta_1, \theta_2, \theta_3, \theta_4, \theta_5\}$ of five states.  We set the signal space of each agent $i \in \mathcal{V}$ to be $\mathcal{S}_i = \{H, T\}$.  For agent 1, we set the distributions of the observations under each of the states as follows:
\begin{align*}
    l_1(H | \theta_1) &= l_1(H | \theta_2) = l_1(H | \theta_5) = 0.5 \\
    l_1(T | \theta_1) &= l_1(T | \theta_2) = l_1(T | \theta_5) = 0.5 \\
    l_1(H | \theta_3) &= l_1(H | \theta_4) = 0.4 \\
    l_1(T | \theta_3) &= l_1(T | \theta_4) = 0.6.
\end{align*}
For each of the other agents, we assign a random permutation of the above distribution over the various states.  

We run the PoE algorithm by setting the true state $\theta^* = \theta_3$, and choosing the parameter $\alpha = \frac{1}{1000}$ for the \texttt{round} function (see Corollary~\ref{cor:finite_time_local_beliefs}).  We show the network belief $\mu_{i,t}$ maintained by a generic agent under this algorithm in Fig.~\ref{fig:network_belief_illustration}. For comparison, we also show the network beliefs for three different agents under the ``min'' rule from \cite{mitraACC19}, which provides the fastest existing (asymptotic) convergence rate for the distributed hypothesis testing problem.  As we can see from the figure, the network beliefs generated by the PoE algorithm converge to the indicator vector for the true state in finite time (approximately 150 time-steps).  Furthermore, since each agent only transmits a binary vector of length $5$ to its neighbors at each time-step, each agent transmits only approximately $750$ bits of information by the time all agents learn the true state.  


\section{Conclusions and Extensions}
In this paper, we first showed that existing algorithms that provide asymptotic learning guarantees can be easily modified to provide finite-time learning; in the context of existing work that seeks to optimize the asymptotic rate of learning, this simple insight indicates that arbitrarily large (asymptotic) rates of learning are easily achievable. 

We next provided a simple algorithm that allows all agents to learn the true state in finite time, and only requires each agent to transmit a binary vector (of length equal to the number of hypotheses) at each time-step.  We followed up this algorithm with a modification that also enables all agents to stop transmitting after a finite length of time, under the assumption that all agents know the diameter of the network.

The key to our approach is that each agent simply leverages its local signals to rule out certain hypotheses, and then the agents run a simple distributed set intersection protocol to find the state that has not been ruled out by every agent. We expect that our algorithm can be readily extended in various directions. For example, for certain classes of time-varying networks, we expect our PoE algorithm will also guarantee finite-time learning, provided that the network is connected over appropriately defined intervals.  Similarly, if the observations at each agent are not i.i.d. over time, one can replace the iterative Bayes rule \eqref{eqn:Bayes} with a non-iterative Bayesian update (where the local belief is updated as a function of all previous measurements at each time-step).  In cases where some of the agents are adversarial, we expect we can make our algorithms resilient by introducing a  ``local-filtering'' step into the distributed set intersection portion of the algorithms, following similar ideas to \cite{mitraACC19}.  Finally, since the convergence properties of our algorithm are essentially dictated by the behavior of the local beliefs at each agent, we expect that one can perform a finite time analysis in order to obtain crisp probabilistic bounds on the time taken for the algorithm to converge.  

As noted in Section~\ref{sec:finite_time_learning}, it will also be of interest to revisit existing asymptotic learning algorithms to understand their performance when modified to yield finite time learning as in Proposition~\ref{prop:finite_time_learning}.  Indeed, as can be observed from the results of our simulation in Fig.~\ref{fig:network_belief_illustration}, modifying the ``min'' rule from \cite{mitraACC19} in this manner would cause the beliefs to converge in less time than required by the PoE algorithm.  This merits a formal analysis and comparison of these existing algorithms.

\bibliographystyle{IEEEtran}
\bibliography{refs}

\end{document}